\documentclass[a4paper,UKenglish,cleveref,autoref]{lipics-v2021}

\title{Improved Approximation and Hardness Results for the Monitoring Edge-Geodetic Set Problem} 

\titlerunning{Improved Approximation and Hardness Results for \megset} 

\author{Davide Bilò}{Department of Information Engineering, Computer Science and Mathematics (DISIM), University of L'Aquila, L'Aquila, Italy}{davide.bilo@univaq.it}{0000-0003-3169-4300}{}

\author{Giordano Colli}{Department of Information Engineering, Computer Science and Mathematics (DISIM), University of L'Aquila, L'Aquila, Italy}{giordano.colli@graduate.univaq.it}{0009-0000-8361-2333}{}

\author{Stefano Leucci}{Department of Information Engineering, Computer Science and Mathematics (DISIM), University of L'Aquila, L'Aquila, Italy}{stefano.leucci@univaq.it}{0000-0002-8848-7006}{} 

\authorrunning{D. Bilò, G. Colli, and S. Leucci}

\Copyright{Davide Bilò, Giordano Colli, and Stefano Leucci}



\EventEditors{John Q. Open and Joan R. Access}
\EventNoEds{2}
\EventLongTitle{42nd Conference on Very Important Topics (CVIT 2016)}
\EventShortTitle{CVIT 2016}
\EventAcronym{CVIT}
\EventYear{2016}
\EventDate{December 24--27, 2016}
\EventLocation{Little Whinging, United Kingdom}
\EventLogo{}
\SeriesVolume{42}
\ArticleNo{23}

\usepackage{xspace}
\usepackage{cleveref}

\newcommand{\classNP}{\ensuremath{\mathsf{NP}}\xspace}
\newcommand{\classAPX}{\ensuremath{\mathsf{APX}}\xspace}
\newcommand{\classZTIME}{\ensuremath{\mathsf{ZTIME}}\xspace}
\newcommand{\classDTIME}{\ensuremath{\mathsf{DTIME}}\xspace}
\newcommand{\polylog}{\operatorname{polylog}}

\newcommand{\megset}{\textnormal{\textsc{MEG-Set}}\xspace}

\newcommand{\symmlabelcover}{\textnormal{\textsc{Symmetric Label Cover}}\xspace}

\newcommand{\vcp}{\textnormal{\textsc{VC3P}}\xspace}

\newcommand{\cost}{\operatorname{cost}}
\newcommand{\E}{\mathbb{E}}

\nolinenumbers
\hideLIPIcs

\begin{document}

\maketitle

\begin{abstract}
We study the \emph{Monitoring Edge-Geodetic Set} problem (\megset), which asks for a minimum-cardinality set $M$ of vertices of an undirected graph $G$ of $n$ vertices such that every edge of $G$ belongs to all shortest paths between some pair of vertices in $M$.

We present new approximation and hardness results for \megset. On the positive side, we design a deterministic polynomial-time $O(n^{1/3}\log^{2/3}n)$-approximation algorithm, improving upon the previously best-known $O(\sqrt{n\log n})$ approximation ratio. On the negative side, for every constant $\varepsilon \in (0,1)$, we prove that the problem admits no polynomial-time $\big(2^{\log^{1-\varepsilon} n}\big)$-approximation algorithm, unless $\classNP\subseteq\classDTIME(n^{\polylog n})$. This substantially improves upon the previously known logarithmic inapproximability bound. We also prove that \megset is $\classNP$-hard on planar $2$-degenerate graphs, whereas $\classNP$-hardness was previously known for $1$-apex $2$-degenerate graphs.

We conclude by observing a connection between monitoring edge-geodetic sets and subsetwise distance preservers, i.e., subgraphs that preserve distances between all pairs of vertices in $S\subseteq V(G)$. Specifically, we show that any upper bound on the size of the latter yields a lower bound on the size of the former.
\end{abstract}

\section{Introduction}
We study the problem of computing a minimum-cardinality \emph{monitoring edge-geodetic set} of an input graph $G=(V,E)$ with $n$ vertices, 
where $M \subseteq V$ is a monitoring edge-geodetic set if every edge of $G$ belongs to all shortest paths between some pair of vertices in $M$. We refer to this optimization problem as \megset.

This problem was introduced in \cite{Foucaud} in the context of network monitoring: the graph $G$ represents a communication network and one wishes to equip a subset of hosts (represented by vertices) with probes that communicate with one-another and measure their pairwise distances.  These probles must be suitably placed to detect link failures, i.e., the removal of any edge $e$ of $G$ must cause the distance between some pair of probes to increase. 
In other words, $e$ is required to lie on all the shortest paths between at least one pair probes, and the objective is to minimize the number of necessary probes.

\megset has received considerable attention, especially in the case of restricted graph classes, as we discuss in the sequel,
while for general graphs it is known to be $\classAPX$-hard
and that, unless $\classNP \subseteq \classZTIME(n^{O(\log \log n)}) $, all polynomial-time  approximation algorithms must have an approximation ratio in $\Omega(\log n)$.
On the positive side, a simple approximation algorithm that reduces the problem to a set-cover can be used to obtain a $O(\sqrt{n \log n})$-approximate solution in polynomial time.
For the special case of planar graphs, the problem is $O(n^{1/4} \sqrt{\log n})$-approximable altough it is still open whether \megset can be solved in polynomial-time in such graphs.

\subparagraph*{Our results.} In this paper we show that \megset is $O(n^{\frac{1}{3}} \log^{\frac{2}{3}} n)$-approximable in polynomial time, improving over the only known approximation algorithm for genral graphs, which has an approximation ratio of  $O(\sqrt{n \log n})$.

On the negative side, we prove that, unless $\classNP \subseteq \classDTIME(n^{\polylog n})$, \megset does not admit any polynomial-time $(2^{\log^{1-\varepsilon} n})$-approximation algorithm, for any constant $\varepsilon \in (0, 1)$.
To the best of our knowledge, this is the first result ruling out the existence of polynomial-time approximation algorithms with polylogarithmic approximation factors (under the above assumption).

We also show that exact polynomial-time algorithms are unlikely to exist even in the special case of $2$-degenerate planar graphs, since the problem remains $\classNP$-hard. This answers a question posed in \cite{BCFL2025}, and improves over the previously known $\classNP$-hardness result for $2$-degenerate $1$-apex graphs \cite{BCFL2025arxiv}.

Finally, we link the size of \megset solutions to the number of edges needed to build a subsetwise distance preserver of $G$. This connection yields a novel lower bound on the size of monitoring edge-geodetic sets as a function of the number of vertices and edges of $G$. 

\subparagraph*{Related work.}

The minimum monitoring edge-geodetic set problem was introduced in \cite{Foucaud}, where the authors studied
the size of such sets in general graphs and in special graph classes including
trees, cycles, unicyclic graphs, complete graphs, grids, and
hypercubes. Subsequent work established bounds on the size of monitoring edge-geodetic sets for graph products.
In particular, \cite{haslegrave} considered Cartesian and
strong products, while \cite{XuYBZS24} studied joins and direct products.
Bounds that are function of the graph's girth and chromatic number have been given in \cite{FMMSST2025}. 
Further results concern convex polytopes and
ladder, butterfly, circulant, and Benes networks \cite{TanLWL23},
as well as radix triangular mesh networks and Sierpi\'nski
graphs~\cite{MaJYL24}.

The complexity status of the problem was studied by \cite{haslegrave} which proved that \megset is $\classNP$-hard on general graphs.
Such hardness result holds even in for the special case of
$2$-degenerate $1$-apex graphs \cite{BCFL2025arxiv}, while the problem is $\classAPX$-hard on $3$-degenerate graphs \cite{BCFL2025arxiv,FoucaudMSST26} and not approximable within any constant better than $2$ if the unique games conjecture holds \cite{BCFL2025arxiv}.\footnote{A very recent preprint claims to have proven the unique games conjecture \cite{OAI:Unique-Games-Theorem-2026}.}
On general graphs, all polynomial-time approximation algorithms for \megset must have an approximation ratio in $\Omega(\log n)$, unless $\classNP \subseteq \classZTIME(n^{O(\log \log n)})$ \cite{BCFL2025,BCFL2025arxiv}.\footnote{In \cite{BCFL2025}, the latter condition was stated incorrectly. The full version \cite{BCFL2025arxiv} gives the correct statement.}
Moreover, if the exponetial time hypothesis \cite{ImpagliazzoPZ01} holds, \megset cannot be solved in time $2^{o(n)}$ on sparse $2$-degenerate graphs, nor in $2^{o(\sqrt{n})}$ time on sparse $2$-degenerate $1$-apex graphs \cite{BCFL2025arxiv}.

On the positive side,
polynomial-time algorithms are known for interval graphs,
distance-hereditary graphs, $P_4$-sparse graphs, bipartite permutation
graphs, strongly chordal graphs~\cite{FoucaudPP26}, and chordal graphs \cite{Marcille26}.
The problem also admits FPT algorithms for general graphs parameterized
by clique-width plus diameter, and for chordal graphs
parameterized by treewidth~\cite{FoucaudMSST26}.

As far as approximation algorithms are concerned, a reduction to \textsc{Set Cover} gives polynomial-time solutions of size $O(k^2\log n)$, where $k$ is the optimum, and hence an $O(\sqrt{n\log n})$-approximation \cite{Colli2023,BiloCF024,FoucaudMSST26}.
Approximation algorithms are also known for graph classes that admit efficiently computable balanced vertex separators \cite{BCFL2025}.
As a consequence, \megset is $O(n^{1/4}\sqrt{\log n})$-approximable on $K_h$-minor-free graphs with $h = O(1)$, on bounded-genus graphs, and on $k$-apex graphs for $k=O(n^{1/4})$. Moreover, it is $O(\log^{3/2}n)$-approximable on graphs with bounded treewidth.

\subparagraph*{Paper organization}
We start by giving some preliminary definition in \Cref{sec:preliminaries}.
\Cref{sec:megset-approximation} presents our approximation algorithm for general graphs, while our inapproximability result is given in \Cref{sec:inapproximability}.
\Cref{sec:megset-planar-hardness} proves the $\classNP$-hardness of \megset on planar graphs. 
\Cref{sec:preservers} concludes the paper and shows how lower bounds on the size of monitoring edge-geodetic can be obtained from upper bounds on subsetwise preservers.

\section{Preliminaries}
\label{sec:preliminaries}
We use $\log z$ (resp.\ $\ln z$) to denote the binary (resp.\ natural) logarithm of $z \in \mathbb{R}^+$.
All graphs used throughout the paper are finite, simple, and
undirected. Moreover, the graphs are unweighted unless otherwise specified.
Let $G=(V,E)$ be a possibly weighted graph, and let
$n=|V|$ and $m=|E|$. For $u,v\in V$, we denote by
$d_G(u,v)$ the distance between $u$ and $v$ in $G$, and by
$\mathcal{P}_G(u,v)$ the set of all shortest paths from $u$ to $v$ in $G$.
When the underlying graph is clear from the context, we simply write
$d(u,v)$ and $\mathcal{P}(u,v)$.

For two distinct vertices $x,y\in V$, let $ \mathcal{M}_G(x,y)=\bigcap_{P\in\mathcal{P}_G(x,y)} E(P)$
be the set of edges monitored by the pair $\{x,y\}$. Thus, an edge
$e\in E$ is monitored by $\{x,y\}$ if and only if
$e\in\mathcal{M}_G(x,y)$. When the underlying graph is clear, we
simply write $\mathcal{M}(x,y)$.
For a set $M\subseteq V$ and an edge $e\in E$, we define the
\emph{monitoring set} of $e$ with respect to $M$ as the set of pairs
of vertices in $M$ that monitor $e$, that is, $
    \mathcal{M}_M(e)
    =
    \left\{
        \{x,y\}\in\binom{M}{2} :
        \{x,y\}\text{ monitors }e
    \right\}$.

\subparagraph*{Monitoring edge-geodetic sets.}
A set $M\subseteq V$ is a \emph{monitoring edge-geodetic set} if,
for every $e\in E$, there is at least one pair of vertices in $M$
that monitors $e$, that is $\mathcal{M}_M(e)\neq\emptyset$ for all $e\in E$.
Given an unweighted graph $G$, the \megset problem asks to find a monitoring edge-geodetic set with the minimum number of vertices.

Without loss of generality, we assume that all instances of \megset are connected graphs with at least one edge.\footnote{Disconnected instances can be decomposed into one sub-instance per connected component. A graph with no edges trivially has the empty set as its only optimal solution.}

\section{A polynomial-time \texorpdfstring{$\widetilde{O}(n^{1/3})$}{O~(n\textasciicircum{}(1/3))}-Approximation Algorithm for \megset}
\label{sec:megset-approximation}

In this section we give a $\widetilde{O}(n^{1/3})$-approximation algorithm for \megset.

We start by presenting a randomized algorithm based on rounding the solution of a relaxed ILP that returns a set $M$ such that: (i) $M$ is a monitoring edge-geodetic set with positive constant probability, and (ii) $M$ has expected size $\widetilde{O}(\sqrt{n |M^*|})$, where $M^*$ is a minimum-cardinality monitoring edge-geodetic set.
We then derandomize this algorithm using the method of conditional expectation, obtaining a deterministic algorithm returning a solution of size $\widetilde{O}(\sqrt{n |M^*|})$.

The ILP formulation resembles the one used to provide a similar randomized approximation algorithm for the \textsc{Min-Rep} problem in~\cite{charikarmoses}. However, our analysis exploits structural properties specific to \megset, which are crucial for derandomizing the algorithm without loss of quality.Derandomizing the algorithm of \cite{charikarmoses} without asymptotic loss of quality remains an open problem, and the best determinsitic algorithm is given in \cite{TirodkarV17}.

Our deterministic algorithm is then combined with the known $\widetilde{O}(|M^*|)$ approximation of \megset  \cite{Colli2023,BiloCF024,FoucaudMSST26}, to obtain the claimed result.

\subparagraph*{Algorithm description.}
We can think of a solution to \megset as a set $M$ of vertices, together with a \emph{witness} for each edge $e$, namely some vertex pair in $\binom{M}{2} \cap \mathcal{V}_v(e)$. Notice that, while there might be multiple pairs in $M^2$ that monitor $e$, we require exactly one to be selected as $e$'s witness. 
Consider the following ILP associated with an instance of \megset:

\begin{equation*}
\begin{aligned}
  \text{minimize}\qquad
  &\sum_{v\in V}x_v\\
  \text{subject to}\qquad
  &\sum_{\{a,b\}\in\mathcal{M}_V(e)}w^e_{ab}=1
      && \quad\forall e\in E,\\
  &\sum_{\substack{\{a,b\}\in\mathcal{M}_V(e)\\v\in\{a,b\}}}
      w^e_{ab}\le x_v
      && \quad\forall e\in E,\ v\in V,\\
  &x_v \in \{0,1\}
      && \quad\forall v\in V,\\
  &w^e_{ab} \in \{0,1\}
      && \quad\forall e\in E,\ \{a,b\}\in\mathcal{M}_V(e).
\end{aligned}
\end{equation*}

The above formulation contains a decision variable $x_v$ for each vertex $v \in V$, which encodes whether $v$ is selected in the returned solution. Moreover, for each edge $e$ and all pairs $\{a,b\} \in\mathcal{M}_V(e)$, the decision variable $w_{ab}^e$ encodes whether $\{a,b\}$ is $e$'s witness.
The first constraint ensures that exactly one witness is selected for each edge, while the second constraint ensures that if $\{a,b\}$ is a witness of some edge $e$, then both $a$ and $b$ must be selected in the solution, i.e., the variables $x_a$ and $x_b$ are both set to $1$.

Our algorithm computes an optimal solution
$(\widetilde{x},\widetilde{w})$ to the relaxation of the ILP
obtained by allowing all variables to take values in $[0,1]$. Here,
$\widetilde{x}=(\widetilde{x}_v)_{v\in V}$ is the vector of
vertex-selection variables, while
$\widetilde{w}=(\widetilde{w}_{ab}^e)_{e\in E,\,
\{a,b\}\in\mathcal{M}_V(e)}$ is the vector of witness-selection
variables.

Let $t = 4\ln(4m)$. For a generic vector $x = (x_v)_{v \in V}$ with $x_v \in [0,1]$, we denote with $M(x)$ a set obtained by independently selecting each $v \in V$ with probability $p(x_v) = \min\{\sqrt{t x_v}, 1\}$. 
The algorithm returns $M(\widetilde{x})$.

\subparagraph*{Algorithm analysis.}
We now argue that, with probability at least $\frac{1}{4}$, $M(\widetilde{x})$ is a monitoring edge-geodetic set of $G$ having size $O(\sqrt{n |M^*| \log n})$.

Given a generic assignment $w$ to the variables $w_{ab}^e$ and an edge $e$,
we define an edge-weighted  \emph{auxiliary graph} $G_e(w)$ having  $V$ as its vertex set, and containing an edge $(u,v)$
of weight $w_{ab}^{e}$ for each pair $\{a,b\} \in \mathcal{M}_V(e)$.
We also let $G^+(w)$ be the spanning subgraph of $G_e(w)$ that contains only edges of positive weight.
We start by observing that all $G_e(w)$ (and hence all $G^+_e(w)$) are bipartite.
\begin{lemma}\label{lem:monitoring-pair-graph-bipartite}
All auxiliary graphs are bipartite.
\end{lemma}
\begin{proof}
Fix an edge $e=(u,v)\in E$, and define
$V_u=\{z\in V : d_G(z,u) < d_G(z,v)\}$ and $V_v= V\setminus V_u$.
Consider an edge $(a,b)$ of $G_e$. Since $\{a,b\}$ monitors $e$,
there exists a shortest path $\pi$ from $a$ to $b$ in $G$ that traverses $e$.
If $\pi$ traverses $e$ from $u$ to $v$, then $d_G(a,u)<d_G(a,v)$ and
$d_G(b,v)<d_G(b,u)$, showing that $a\in V_u$ and $b\in V_v$.
Symmetrically, if $\pi$ traverses $e$ from $v$ to $u$, then
$a\in V_v$ and $b\in V_u$.

Then, every edge of $G_e$ has one endvertex in $V_u$ and the other in
$V_v$, hence the claim.
\end{proof}

We now show that there exists an optimal solution $(\widetilde{x}, \overline{w})$ such that all $G_e^+(\overline{w})$ are forests, and that $\overline{w}$ can be found efficiently from $\widetilde{w}$.

\begin{lemma}\label{lem:megset-lp-forest-support}
  Given an optimal solution $(\widetilde{x}, \widetilde{w})$  to the relaxation of the ILP, one can compute in polynomial time an optimal
  solution $(\widetilde{x},\overline{w})$ to the relaxation of the ILP guaranteeing that
  all $G_e^+(\overline{w})$, for $e\in E$, are forests.
\end{lemma}
\begin{proof}
Instead of directly proving the claim, we show the following: if $G_e^+(\widetilde{w})$ contains a cycle,
then there exists a procedure that efficiently finds another vector $\overline{w}$ such that:
(i) $\overline{w}$ can only differ
from $\widetilde{w}$ on the values of the variables $\overline{w}^e_{ab}$ associated with $e$;
(ii) $(\widetilde{x}, \overline{w})$ is a feasible solution of the relaxed ILP;
and (iii) $G_e^+(\overline{w})$ has fewer edges than $G_e^+(\widetilde{w})$. 
Observe that (ii) implies the optimality of $(\widetilde{x}, \overline{w})$ as the objective function only depends on $\widetilde{x}$.
The claim follows by repeatedly applying the above procedure (no more than $O(n^2 |E|)$ times). 

Suppose then that $G_e^+(\widetilde{w})$ contains a cycle $C$.
By \cref{lem:monitoring-pair-graph-bipartite}, $C$ has some even
length $k$. Let $f_1$ be an edge of minimum weight in $C$, let $\delta$ be the weight of $f_1$, and let $f_1, f_2, \dots, f_k$ be the edges of $C$, in some cyclic order starting from $f_1$.
The vector $\overline{w}$ is obtained by subtracting $\delta$ from (resp. adding $\delta$ to) all the variables $\widetilde{w}^e_{ab}$ such that $f_i = (a,b)$ for an odd (resp.\ even) index $i$. Then (i) holds by construction.

Observe that all affected variables remain non-negative. Moreover, since $k$ is even, the number of variables increased by $\delta$ equals the number of variables decreased by $\delta$; hence, their sum in $\widetilde{w}$ coincides with their sum in $\overline{w}$. Therefore, the first constraint of the linear program remains satisfied (and shows that no variable in $\overline{w}$ exceeds $1$).
Regarding the second constraint, observe that if a vertex $v$ is not part of $C$, then the left-hand sum is unaffected. 
If $v$ is part of $C$, then exactly two variables of the left-hand sum are affected: since the two edges of $C$ incident to $v$ have indices of opposite parity, one variable decreases by $\delta$, while the other increases by $\delta$. Hence, the total sum does not change and the second constraint is also satisfied. This proves (ii).

Finally, (iii) holds since $f_1$ is an edge of $G_e^+(\widetilde{w})$ but not of $G_e^+(\overline{w})$, and
only non-zero variables $\widetilde{w}^e_{ab}$ can have a different value in $\overline{w}$.
\end{proof}

\begin{lemma}\label{lem:megset-heavy-star-forest}
Let $e \in E$, and let $(\widetilde{x},\overline{w})$ be an optimal solution to the relaxation of the ILP
such that $G_e^+(\overline{w})$ is a forest. 
There exists a collection $\mathscr{S}_e$ of vertex-disjoint stars contained in
$G_e^+(\overline{w})$, each with at least one edge, such that
$\sum_{S\in\mathscr{S}_e}\sum_{(u,v) \in E(S)}\overline{w}_{uv}^e \ge 1/2$.
\end{lemma}
\begin{proof}
By the first constraint of (the relaxation of) the ILP, the sum of all edge-weights of $G_e^+(\overline{w})$ is $1$.

Root each tree of $G_e^+(\overline{w})$ in an arbitrary vertex.
Partition the edges of $G_e^+(\overline{w})$
into two sets $E_{\text{even}}$ and $E_{\text{odd}}$ according to whether the parent endvertex lies at
an even or an odd level in the containing tree.
Then, the sum of the weights of the edges in at least one set $E^* \in \{ E_{\text{even}}, E_{\text{odd}} \}$ must be at least $\frac{1}{2}$. We select $\mathscr{S}_e$ as the collection of stars induced by $E^*$, where each star is centered in its vertex of lowest depth.
\end{proof}

Fix some $\overline{w}$ such that $(\widetilde{x}, \overline{w})$ is an optimal solution to the relaxation of the ILP
and, for each edge $e$, let $\mathscr{S}_e = \{S_1, \dots, S_k\}$ be a collection of vertex-disjoint star subgraphs of $G_e^+(\overline{w})$, each with at least one edge, such that $\sum_{S\in\mathscr{S}_e}\sum_{(u,v) \in E(S)}\overline{w}_{uv}^e \ge 1/2$. Denote by $c_i$ and by $L_i$ the center and the leaves of $S_i$, respectively.
For a vector $x = (x_v)_{v \in V}$ where $x_v \in [0,1]$, define:
\[
  f_e(x) = \prod_{i=1}^k  \left( 
    (1 - p(x_{c_i})) + p(x_{c_i}) \prod_{v \in L_i} (1- p(x_v)) 
    \right).
\]

Given a set $M$, we say that $e$ is \emph{covered} by a star $S_i \in \mathscr{S}_e$ w.r.t.\ $M$ if there exists at least one edge in $S_i$ having both endpoints in $M$. Notice that if an edge $e$ is covered by some star in $\mathscr{S}_e$ w.r.t. $M$, then $e$ is monitored by $M$.

\begin{lemma}
  Fix an edge $e \in E$ and a vector $x$. The probability that $e$ is not covered by any star in $\mathscr{S}_e$ w.r.t.\ $M(x)$ is at most $f_e(x)$.
\end{lemma}
\begin{proof}
  We say that an edge $(c_i, v)$ of $S_i$ is \emph{selected} if both
  $c_i$ and $v$ are in $M(x)$.
  Observe that $e$ is covered iff at least one edge of some star $S_i$ is selected.
  
  No edge of $S_i$ is selected exactly when one of the following two disjoint events happens: (i) $c_i \not\in M(x)$, or (ii) $c_i \in M(x)$ and $L_i \cap M(x) = \emptyset$. 
  Then, (i) happens with probability $1 - p(x_{c_i})$, and (ii) happens with probability $p(x_{c_i}) \prod_{v \in L_i} (1 - p(x_v))$, where we used the fact that vertices are added indepently to $M(x)$.

  Since the stars $S_1, \dots, S_k$ are vertex-disjoint, the events ``$S_i$ contains a selected edge'' are independent. Hence, the probability that no star contains a selected edge is:
  \begin{equation*}
    \prod_{i=1}^k \left( (1 - p(x_{c_i}))   +  p(x_{c_i}) \prod_{v \in L_i} (1 - p(x_v)) \right)
    = f_e(x). \tag*{\qedhere}
  \end{equation*}
\end{proof}

\begin{lemma}
  \label{lemma:apx_ub_f_e}
  For every $e \in E$, $f_e( \widetilde{x} ) < \frac{1}{4m}$.
\end{lemma}
\begin{proof}
Fix $e\in E$ and use the previously defined collection
$\mathscr{S}_e=\{S_1,\dots,S_k\}$, with centers $c_i$ and leaf
sets $L_i$. For each $i$, let $s_i=\sum_{v\in L_i}\overline{w}_{c_i v}^e$, and observe that $s_i > 0$ and that $\sum_{i=1}^k s_i \ge 1/2$.

To ease the notation, define 
$\gamma_i=1-p(\widetilde{x}_{c_i}) + p(\widetilde{x}_{c_i}) \prod_{v\in L_i}(1-p(\widetilde{x}_v))$
so that
$f_e(\widetilde{x})=\prod_{i=1}^k \gamma_i$.

Since $(\widetilde{x},\overline{w})$ is feasible, the second constraint gives
$\widetilde{x}_{c_i}\ge s_i$ and
$\widetilde{x}_v\ge\overline{w}_{c_i v}^e$ for every $v\in L_i$.
We show that $\gamma_i < e^{-t s_i/2}$.

Suppose first that $p(\widetilde{x}_{c_i})=1$. If some leaf $v\in L_i$
also has $p(\widetilde{x}_v)=1$, then $\gamma_i=0 < e^{-t s_i/2}$. Otherwise,
$p(\widetilde{x}_v) = \sqrt{t \widetilde{x}_v } <1$ for every $v \in L_i$,  and  $p(\widetilde{x}_v)= \sqrt{t\widetilde{x}_v}
 \ge t\widetilde{x}_v\ge t\overline{w}_{c_i v}^e$. Then:
\[
  \gamma_i=\! \prod_{v\in L_i} \! (1-p(\widetilde{x}_v))
  \le\! \prod_{v\in L_i} \! e^{- p(\widetilde{x}_v)}
  =\exp\!\left(\!-\!\!\sum_{v\in L_i}p(\widetilde{x}_v)\right)
  \le \exp\!\left(\!-t\! \sum_{v\in L_i} \overline{w}^e_{c_i v}\right)
  = e^{-t s_i} < e^{-\frac{t s_i}{2}}.
\]

Suppose now that $p(\widetilde{x}_{c_i})<1$. Then
$t s_i\le t\widetilde{x}_{c_i}<1$ and
$p(\widetilde{x}_{c_i}) = \sqrt{t \widetilde{x}_{c_i}} \ge\sqrt{t s_i}$. If some leaf $v\in L_i$ has $p(\widetilde{x}_v)=1$, then:
$ \gamma_i=1-p(\widetilde{x}_{c_i})
  \le e^{-p(\widetilde{x}_{c_i})}
  \le e^{-\sqrt{t s_i}} < e^{-\frac{t s_i}{2}}$.
Otherwise, $p(\widetilde{x}_v) <1$ for all $v \in L_i$ and, from the subadditivity of the square root on nonnegative arguments,
$
  \sum_{v\in L_i}p(\widetilde{x}_v)
  =\sqrt{t}\sum_{v\in L_i}\sqrt{\widetilde{x}_v}
  \ge\sqrt{t}\sum_{v\in L_i}\sqrt{\overline{w}_{c_i v}^e}
  \ge\sqrt{t s_i}$.

Using $p(\widetilde{x}_v) \ge \sqrt{t \overline{w}^e_{c_i v}} > 0$ and the inequality $1-e^{-z}\ge z/2$ for $z\in[0,1]$, we have:
\[
  1-\prod_{v\in L_i}(1-p(\widetilde{x}_v))
  > 1-\prod_{v\in L_i} e^{-p(\widetilde{x}_v)}
  = 1-\exp\left(-\sum_{v\in L_i}p(\widetilde{x}_v)\right)
  \ge 1 - e^{-\sqrt{t s_i}}
  \ge\frac{\sqrt{t s_i}}{2},
\]
which implies:
\[
  \gamma_i
  =1-p(\widetilde{x}_{c_i})
      \left(1-\prod_{v\in L_i}(1-p(\widetilde{x}_v))\right)
    < 1- \sqrt{t s_i} \cdot \frac{\sqrt{t s_i}}{2}
    = 1 - \frac{t s_i}{2}
    \le e^{-t \frac{s_i}{2}}.
\]

Hence, $\gamma_i < e^{-t s_i/2}$ regardless of the case. Multiplying these upper bounds across all stars in $\mathscr{S}_e$ and using $t=4\ln(4m)$ yields:
\[
  f_e(\widetilde{x})
  = \prod_{i=1}^k \gamma_i
  < \prod_{i=1}^k e^{-t s_i/2}
  =\exp \left(-\frac{t}{2}\sum_{i=1}^k s_i\right)
  \le e^{-\frac{t}{4}}
  =\frac{1}{4m}. \tag*{\qedhere}
\]
\end{proof}

\noindent The definition of coverage, \Cref{lemma:apx_ub_f_e}, and the union bound yield the following: 

\begin{corollary}
  \label{cor:apx_success_probability}
  With probability at least $\frac{3}{4}$, $M(\widetilde{x})$ is a monitoring edge-geodetic set of $G$.
\end{corollary}

\noindent Finally, we upper bound the expected size of the returned set $M(\widetilde{x})$:

\begin{lemma}
  \label{lemma:apx_expected_solution_size}
  $\E[|M(\widetilde{x})|] =\sum_{v\in V}p(\widetilde{x}_v) \le \sqrt{n t |M^*|}$.
\end{lemma}
\begin{proof}
The Cauchy-Schwarz inequality shows that, for non-negative $z_1,\dots,z_n$, we have $
  \sum_{i=1}^n\sqrt{z_i}
  \le\sqrt{n\sum_{i=1}^n z_i}$. Together with $|M^*| \ge \sum_{v\in V}\widetilde{x}_v$, this implies:
\[
  \E[|M(\widetilde{x})|]
  =\sum_{v\in V}p(\widetilde{x}_v)
  =\sum_{v\in V}\min\{1,\sqrt{t\widetilde{x}_v}\}
  \le\sum_{v\in V}\sqrt{t\widetilde{x}_v}
  \le\sqrt{nt\sum_{v\in V}\widetilde{x}_v}
  \le\sqrt{nt|M^*|}. \tag*{\qedhere}
\]
\end{proof}

\Cref{cor:apx_success_probability} and \Cref{lemma:apx_expected_solution_size}, together with Markov's inequality and the union bound already provide a randomized approximation-algorithm for \megset.
In detail, the probability that $|M(\widetilde{x})|$ exceeds $2\sum_{v\in V}p(\widetilde{x}_v)$ is at most $\frac{1}{2}$,
hence the algorithm has a probability of at least $1 - \frac{1}{4} - \frac{1}{2} = \frac{1}{4}$ of returning a monitoring edge-geodetic set of size at most $2\sum_{v\in V}p(\widetilde{x}_v) \le 2 \sqrt{n t |M^*|} = O(\sqrt{n |M^*| \log n})$.

Moreover, since solutions can be checked in polynomial time, and we can test whether $|M(\widetilde{x})| \le 2\sum_{v\in V}p(\widetilde{x}_v)$ holds, we can repeatedly sample $M(\widetilde{x})$ until we encounter a feasible solution that is not too large. This yields a Las Vegas algorithm that performs at most $4$ trials in expectation.

\subparagraph*{Derandomization.}
We now derandomize the above algorithm. Our deterministic version of the algorithm solves the relaxed ILP to obtain an optimal solution $(\widetilde{x}, \widetilde{w})$ and converts it to an optimal solution $(\widetilde{x}, \overline{w})$ such that all $G_e^+(\overline{w})$ are forests, as shown by \Cref{lem:megset-lp-forest-support}.
This allows to find a collection $\mathscr{S}_e$ of stars of $G_e^+(\overline{w})$ satisfying the conditions of \Cref{lem:megset-heavy-star-forest}, for each $e \in E$.

Then, instead of choosing $M(\widetilde{x})$ by independently selecting each vertex $v$ with probability $p(\widetilde{x}_v)$, it examines one vertex at a time and deterministically decides whether such a vertex will be included in the solution using the method of conditional expectations.
Here excluding (resp.\ including) a vertex from (resp.\ into) the solution $v$ can be seen as forcing the corresponding variable $\widetilde{x}_v$ to $0$ (resp.\ $1$) which forces $p(\widetilde{x}_v)$ to become $0$ (resp.\ $1$).

In the rest of this section we provide the details of this iterative selection step.

Fix an ordering $V=\{v_1,\dots,v_n\}$ and write $x_j=x_{v_j}$.
For a vector $x\in\{0,1\}^n$, let $U_e(x)$ be the indicator function that
equals $1$ if no star in $\mathscr{S}_e$ covers $e$ w.r.t.\ $\{v_j \mid x_j = 1\}$.
Define $\nu = \sum_{i=1}^n p(\widetilde{x}_{v_i})$ and
$\varphi(x)=\sum_{e\in E}U_e(x)
      +\frac{3}{4\nu}\sum_{j=1}^n x_j$.
Let $X=(X_1,\dots,X_n)$ be a vector of independent Bernoulli random variables
where, $X_j$ has parameter $p(\widetilde{x}_j)$.
Using the definition of $f_e$, \Cref{lemma:apx_ub_f_e}, \Cref{lemma:apx_expected_solution_size}, and the linearity of expectation:
\[
  \E[\varphi(X)]
  =\sum_{e\in E}\E[U_e(X)]
      +\frac{1}{4\nu}\sum_{j=1}^n\E[X_j]
  \le \sum_{e\in E}f_e(\widetilde{x})+\frac{3}{4}
  < m\cdot \frac{1}{4m} + \frac{3}{4}= 1.
\]

Consider some $i=0, \dots, n$ and some choice of $x_1, \dots, x_i \in \{0, 1\}$
with the constraint that $x_j = p(\widetilde{x}_j)$ whenever $p(\widetilde{x}_j) \in \{0, 1\}$.
Denote by $x^{(i)}$ the vector obtained from $\widetilde{x}$ by replacing $\widetilde{x}_j$ with $x_j$ for $j=1,\dots, i$.
Let $H_i$ denote the event $\{X_1=x_1,\dots,X_i=x_i\}$, with $H_0$ being the certain event.
Then, the following conditional expectation can be evaluated in polynomial time:
\[
  \E[\varphi(X)\mid H_i]
  =\sum_{e\in E}f_e(x^{(i)})
      +\frac{3}{4\nu}\sum_{j=1}^n p(x^{(i)}_j),
\]

Our deterministic algorithm works iteratively for $i=0, \dots, n-1$. When $i$ is considered, it permanently fixes variable $x_{i+1}$ to a value in $\{0, 1\}$ such that
  $\E[\varphi(X)\mid H_i,X_{i+1}=x_{i+1}]
  \le\E[\varphi(X)\mid H_i]$. 
In detail, if $p(\widetilde{x}_{i+1}) \in \{0, 1\}$, we simply choose  $x_{i+1} = p(\widetilde{x}_{i+1})$.
Otherwise, the law of total expectation gives:
\[  \E[\varphi(X)\mid H_i]
  =(1-p(\widetilde{x}_{i+1})) \cdot
       \E[\varphi(X)\mid H_i,X_{i+1}=0]
     +p(\widetilde{x}_{i+1}) \cdot 
       \E[\varphi(X)\mid H_i,X_{i+1}=1].
\]
Hence, at least one of the two conditional expectations on the
right-hand side is at most $\E[\varphi(X)\mid H_i]$.
We evaluate both and choose the value of $x_{i+1}$ that yields the
smaller conditional expectation, breaking ties arbitrarily.

Once all variables are fixed, we obtain a vector $x\in\{0,1\}^n$
satisfying $\varphi(x)=\E[\varphi(X)\mid X=x] \le\E[\varphi(X)] < 1$.
Since every term of $\varphi(x)$ is nonnegative and each $U_e(x)$
belongs to $\{0,1\}$, this implies $U_e(x)=f_e(x)=0$ for every
$e\in E$ and $\sum_{j=1}^n x_j \le \frac{4}{3} \nu \le \frac{4}{3} \sqrt{nt |M^*|}$, as shown by \Cref{lemma:apx_expected_solution_size}.
Thus, the set $M(x)=\{v_j:x_j=1\}$ has size at most $\frac{4}{3} \sqrt{nt |M^*|} =\frac{8}{3} \sqrt{n |M^*|  \ln (4m) }$ and each edge $e \in E$ is covered by some star w.r.t.\ $M(x)$, i.e., $M(x)$ is a monitoring edge-geodetic set of $G$.
We have therefore shown the following:
\begin{lemma}
\label{lemma:apx_derandomized}
There exists a deterministic polynomial-time algorithm for \megset that returns a solution of size at most $\frac{8}{3} \sqrt{n |M^*|  \ln (4m)}$, where $M^*$ is an optimal solution.
\end{lemma}

\subparagraph*{The final approximation algorithm.}
By returning the smallest solution between the one provided by $O\left(\sqrt{\frac{n \log n}{|M^*|}} \right)$-approximation algorithm of \Cref{lemma:apx_derandomized}, and the one resulting from the known $O(|M^*| \log n)$-approximation algorithm based on set-cover \cite{Colli2023,BiloCF024,FoucaudMSST26}, we obtain the main result of this section:
\begin{theorem}
There exists a deterministic polynomial-time $O(n^{1/3} \log^{2/3} n)$-approximation algorithm for \megset.
\end{theorem}
\begin{proof}
  Fix some intger $n \ge 2$ and, for $x \in \mathbb{R}^+$, let $f(x) =  \min\left\{ \sqrt{\frac{n \log n}{x}}, x \log n \right\}$.
  The approximation ratio achieved by the above algorithm is at most $c \cdot f(|M^*|)$, for some absolute constant $c>0$.
  Since one term of the minimum is continuous and increasing in $x$, while the other term is continuous and decreasing in $x$, the choice $x^*$ of $x$ that maximizes $f$ is $x^* = ( \frac{n}{\log n} )^{\frac{1}{3}}$, as it satisfies $\sqrt{\frac{n \log n}{x^*}} = x^* \log n$. Therefore, $f(x) \le f(x^*) = n^{1/3} \log^{2/3} n$.
\end{proof}

\section{Inapproximability of \megset}
\label{sec:inapproximability}
In this section, we give a reduction
from \symmlabelcover showing that, for every constant
$\varepsilon \in (0,1)$, \megset admits no polynomial-time
$\left(2^{\log^{1-\varepsilon} n}\right)$-approximation algorithm, unless $\classNP\subseteq \classDTIME(n^{\polylog n})$.

An instance $\mathcal{I}$ of \symmlabelcover consists
of a complete bipartite graph
$\mathcal{H}=(U\cup W,E)$ with $|U|=|W|=q\ge1$,
two finite label sets $A$ and $B$, and a nonempty
relation $R_{u,w}\subseteq A\times B$ for every
$(u,w)\in E$.
A feasible solution assigns a set $\phi(u)\subseteq A$
to every $u\in U$ and a set $\psi(w)\subseteq B$ to
every $w\in W$. For every $e=(u,w)$ in $E$, there
must exist $a\in\phi(u)$ and $b\in\psi(w)$ such that
$(a,b)\in R_{u,w}$. The goal is to minimize
$\cost(\phi,\psi)=
\sum_{u\in U}|\phi(u)|+\sum_{w\in W}|\psi(w)|$.
For $s=\max\{|A|,|B|,q\}$, it is known that \cite{BlueRed}:

\begin{theorem}[\cite{BlueRed}]
\label{thm:symmetric-label-cover-hardness}
For every constant $\delta \in (0,1)$, \symmlabelcover
admits no polynomial-time
$\left(2^{\log^{1-\delta} s}\right)$-approximation algorithm,
unless
$\classNP\subseteq
\mathrm{DTIME}(n^{\polylog n})$.
\end{theorem}

Starting from $\mathcal{I}$, we first construct a
graph $H$ and then combine several copies of $H$
to obtain an instance $G$ of \megset.

The graph $H$ is built as follows. 
For each $e=(u,w)\in E$, we create
vertices $x_e$ and $y_e$ and the edge
$t_e=(x_e,y_e)$. Let $T$ contain the edges $t_e$
for all $e\in E$. The edge $t_e$ will be monitored if
at least one pair of labels assigned to
$u$ and $w$ is in $R_{u,w}$.

For every $u\in U$ and every $a\in A$, we introduce
a vertex $L_{u,a}$ representing the assignment of
label $a$ to $u$. Similarly, for every $w\in W$
and every $b\in B$, we introduce a vertex $R_{w,b}$
representing the assignment of label $b$ to $w$.

For each $e=(u,w)\in E$, we add a path
$P(L_{u,a},x_e)$ of length $5$ for every $a\in A$
and a path $P(y_e,R_{w,b})$ of length $5$ for every
$b\in B$. For every pair $(a,b)\in A\times B$
that does not belong to $R_{u,w}$, we also add a
path $P(L_{u,a},R_{w,b})$ of length $5$.
This provides a shorter connection between vertices representing
labels assigned to $u$ and $v$ that are not in $R_{u,w}$,
preventing such pairs from monitoring $t_e$.
\Cref{fig:label-cover-core} show three constraints
of an instance $\mathcal{I}$ and the corresponding part of $H$.

\begin{figure}
  \centering
  \includegraphics[width=0.90\textwidth]{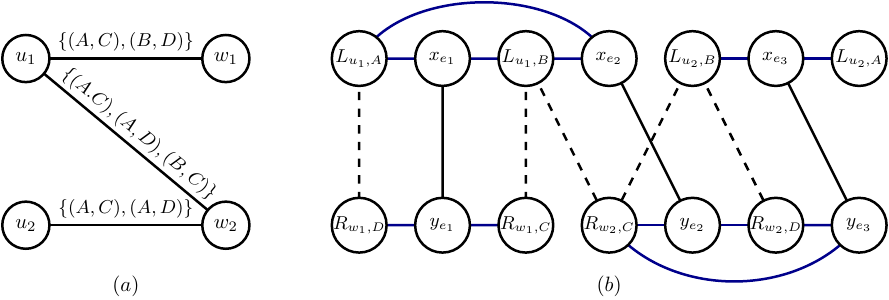}
  \caption{
Part of the construction of $H$ from an
instance $\mathcal{I}$.
(a) The instance has $U=\{u_1,u_2\}$ and
$W=\{w_1,w_2\}$, with possible labels $\{A,B\}$
for vertices in $U$ and $\{C,D\}$ for vertices in $W$.
The three displayed edges are $e_1=(u_1,w_1)$,
$e_2=(u_1,w_2)$, and $e_3=(u_2,w_2)$, each
annotated with its allowed label pairs.
The edge $(u_2,w_1)$ is omitted for clarity.
(b) The corresponding part of $H$.
Solid black lines represent the edges
$t_{e_1}$, $t_{e_2}$, and $t_{e_3}$.
Blue lines represent paths of length five joining
vertices representing labels to the corresponding
vertices $x_e$ and $y_e$.
Dashed black lines represent paths of length five
joining vertices representing label pairs $(a,b)\notin R_{u,w}$.}
  \label{fig:label-cover-core}
\end{figure}

We now construct $G$. Let $h=|V(H)|$, set $r=4h$,
and write $[r]$ for the set of integers from one to $r$.
Add $k=r^2$ vertex-disjoint copies of $H$ to $G$,
arranged in $r$ rows and $r$ columns.
For every $i,j \in [r]$, we denote the corresponding
copy of $H$ by $H^{i,j}$. For each vertex $v$ of $H$,
its copy in $H^{i,j}$ is denoted by $v^{i,j}$.
Likewise, the copy of a path $P(a,b)$ is denoted
by $P(a^{i,j},b^{i,j})$.
Let $t_e^{i,j}$ denote the copy of $t_e$ in
$H^{i,j}$, and let $T_G$ contain these edges for
all $e\in E$ and $i,j \in[r]$, i.e. $T_G = \{t_e^{i,j} :i,j\in [r], e \in E\}$.

Let $V_0$ consist of the vertices $x_e$ and $y_e$
for every $e\in E$, the vertices $L_{u,a}$ for
every $u\in U$ and $a\in A$, and the vertices
$R_{w,b}$ for every $w\in W$ and $b\in B$.
Thus, $V_0$ contains the endpoints of the paths
added to $H$, but excludes their internal
vertices.

To complete our construction of $G$, we consider
every $v\in V_0$ and every $i\in[r]$, and we
add a vertex $\rho_{v,i}$ and add a path
$P(\rho_{v,i},v^{i,j})$ of length $3$ for every
$j\in[r]$. Similarly, for every $v\in V_0$ and
every $j\in[r]$, we add a vertex
$\kappa_{v,j}$ and add a path
$P(\kappa_{v,j},v^{i,j})$ of length $3$ for
every $i\in[r]$. These vertices connect the
copies of $v$ along each row and each column.

For each $v\in V_0$, let $\Gamma(v)$ contain
all vertices $\rho_{v,i}$ and $\kappa_{v,j}$,
as $i$ and $j$ range over $[r]$.
Let $\Gamma$ be the union of these sets over
all $v\in V_0$. For every $g\in\Gamma$, we append
a leaf $\ell(g)$ to $g$, and let $F$ be the set
of these leaves.

Let $D$ contain all pairs $\{x,y\}$
of distinct vertices of $V_0$ for which some
edge of $T$ is moitored by $\{x,y\}$.
For every $\{x,y\}\in D$, every $p\in\Gamma(x)$,
and every $q\in\Gamma(y)$, we add a path $P(p,q)$
of length $6$. These paths serve to prevent
pairs of leaves in $F$ from monitoring the
edges in $T_G$.
This completes the construction of $G$.
\Cref{fig:label-cover-full} shows the
connections between the copies of $H$, the
vertices of $\Gamma$, and their leaves.

\begin{figure}
  \centering
  \includegraphics[width=0.5\textwidth]{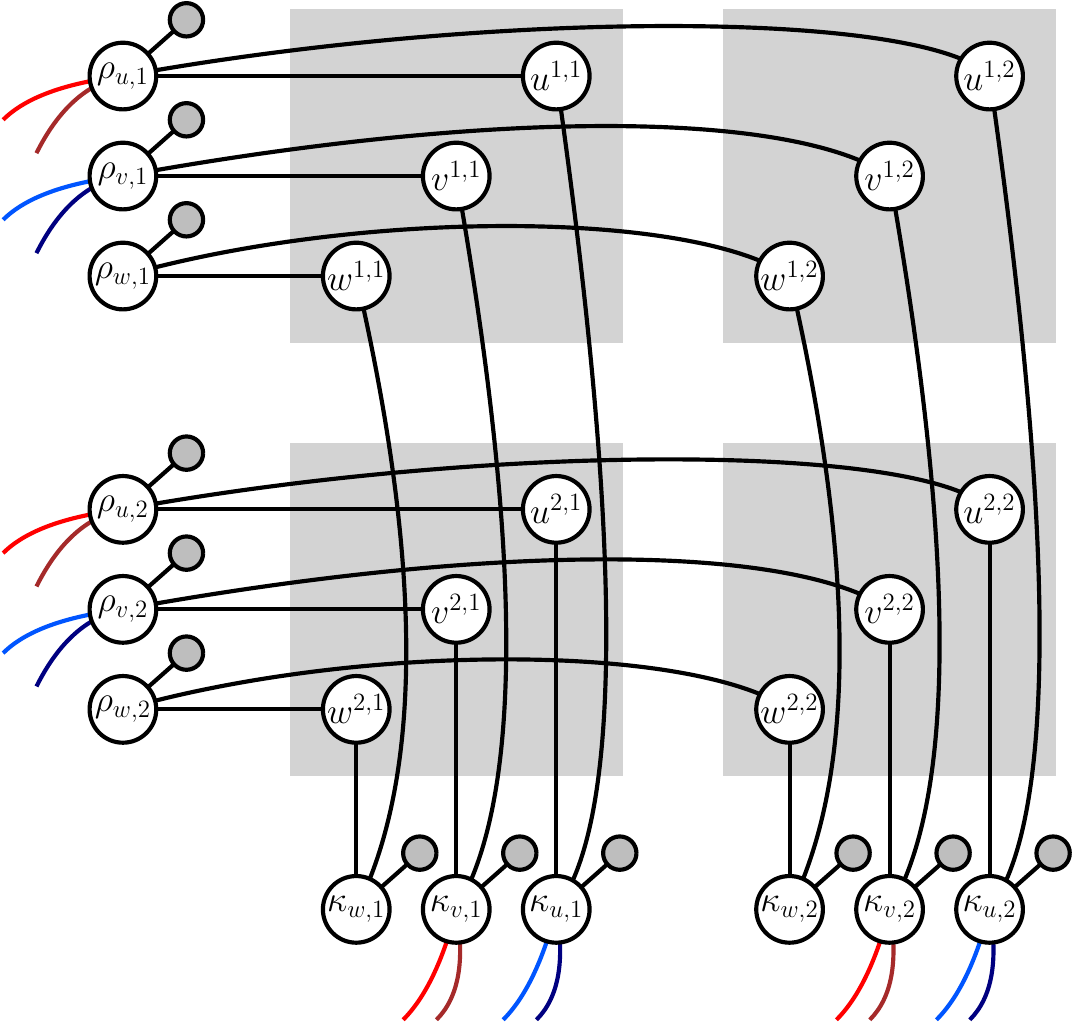}
  \caption{
    Qualitative depiction of the graph $G$.
    The shaded squares represent the copies
    $H^{1,1}$, $H^{1,2}$, $H^{2,1}$, and $H^{2,2}$,
    showing only the copies of three vertices
    $u,v,w\in V_0$.
    For each $z\in\{u,v,w\}$, the black lines joining
    $\rho_{z,i}$ to $z^{i,j}$ and $\kappa_{z,j}$
    to $z^{i,j}$ represent paths of length three.
    The small grey vertices are the leaves
    $\ell(g)\in F$, each joined to its corresponding
    vertex $g\in\Gamma$ by a single edge.
    Assuming $\{u,v\}\in D$, the red and brown curves
    indicate portions of the paths of length six
    between $\rho_{u,i}$ and $\kappa_{v,j}$;
    the blue curves indicate portions of those
    between $\rho_{v,i}$ and $\kappa_{u,j}$.
    Only the ends of these coloured paths are shown.
    The remaining copies, the edges inside each
    copy, the internal vertices of the paths, and
    the other paths between vertices of $\Gamma$
    are omitted for clarity.
  }
  \label{fig:label-cover-full}
\end{figure}

The next two lemmas establish the correctness of the
reduction. \Cref{lem:all-edges-monitored-leaves} shows
that the leaves in $F$ monitor every edge outside $T_G$,
while \cref{lem:normalize-meg-set} shows that any
monitoring edge-geodetic set can be enlarged by a
constant factor to include a compatible pair of labels
for every edge in $T_G$.
Together, they allow us to transfer approximation
guarantees from \megset to \symmlabelcover (up to constant factors).

\begin{lemma}\label{lem:all-edges-monitored-leaves}
Every edge in $E(G) \setminus T_G$ is monitored
by a pair of vertices in $F$.
\end{lemma}
\begin{proof}
Every path between distinct vertices of $\Gamma$
has length at least $6$. Consequently, a path
between two vertices of $\Gamma$ that traverses
another vertex of $\Gamma$ has length at least
$12$.

Fix $v\in V_0$ and $i,j\in[r]$.
The concatenation of $P(\rho_{v,i},v^{i,j})$ and
$P(\kappa_{v,j},v^{i,j})$ forms a path
of length $6$ between $\rho_{v,i}$ and
$\kappa_{v,j}$, while our construction adds no path
between these two vertices of $\Gamma(v)$.
Moreover, $v^{i,j}$ is the only copy of $v$
joined to both by the added paths of length $3$.
Any alternative path that does not contain vertices in
$\Gamma$, must include at least two paths of length $3$  
and at least one edge inside a copy of $H$,
and therefore has length at least $7$.
Therefore, the leaves $\ell(\rho_{v,i})$ and $\ell(\kappa_{v,j})$
monitor every edge of the path obtained by
concatenating  $(\ell(\rho_{v,i}),\rho_{v,i})$,
with $P(\rho_{v,i},v^{i,j})$,
$P(v^{i,j},\kappa_{v,j})$, and 
$(\kappa_{v,j}, \ell(\kappa_{v,j}))$.

Next, fix $\{x,y\}\in D$, $p \in \Gamma(x)$,
and $q \in \Gamma(y)$.
The path $P(p,q)$ has length $6$.
Any alternative avoiding other vertices of
$\Gamma$ concatenates two paths of length three and
at least one edge inside a copy of $H$,
and therefore has length at least $7$.
An alternative shortest path that contains a vertex in 
$\Gamma$ has length at least $12$.
Thus, $P(p,q)$ is the unique shortest path, and
$\ell(p)$ and $\ell(q)$ monitor all its edges.
Together with the preceding case, this shows that
every edge of $G$ outside the copies of $H$ is
monitored by some pair of vertices in $F$.

It remains to consider the edges inside the
copies of $H$. Every such edge outside $T_G$
belongs to some copy of a
$P(s,t)$ of length $5$. The path $P(s,t)$ is the unique shortest path between its endvertices $s,t \in H^{i,j}$, since any alternative path containing vertices in $V(G) \setminus V(H^{i j})$ has length at least $6$, observing that it contains $2$ paths of length at least $3$ to exit and return into $H^{i,j}$.
Moreover, every alternative $s-t$ path that contains only vertices in $V(H^{i j})$ contains at least a path of length $5$ toghether with an edge.
As this path contains no edge of $T$, we have
$\{s,t\}\notin D$.

For every $i,j\in[r]$, the path from
$\rho_{s,i}$ to $\kappa_{t,j}$ that contains
$P(s^{i,j},t^{i,j})$ has length $11$.
Any alternative path including a vertex of
$\Gamma$ has length at least $12$.
Next, observe that no path was added directly between these
two vertices of $\Gamma$, since
$\{s,t\}\notin D$.
Finally, if alteranive shortest path not including other vertices of
$\Gamma$ existed, it would need to use the same two paths
of length $3$ and a path of length $5$ inside
$H^{i,j}$, however their concatenation is not a shortest path.

Thus, $\ell(\rho_{s,i})$ and
$\ell(\kappa_{t,j})$ monitor every edge of
$P(s^{i,j},t^{i,j})$.
This applies to all introduced paths in $H$,
including those connecting incompatible pairs,
and completes the proof.
\end{proof}

\begin{lemma}\label{lem:normalize-meg-set}
Given a monitoring edge-geodetic set $M$ of $G$,
one can compute in polynomial time a set
$M'\subseteq V(G)$ such that:
(i) $M'$ is a monitoring edge-geodetic set of $G$;
(ii) $|M'|\le5|M|$; and
(iii) for every $e=(u,w)\in E$ and every $i,j\in[r]$,
there exists $(a,b)\in R_{u,w}$ such that
$L_{u,a}^{i,j},R_{w,b}^{i,j}\in M'$.
\end{lemma}
\begin{proof}
We first construct $M'$.
Start with $M$ and add both endpoints of every
introduced path containing an internal vertex
of $M$. Call this set $S$.
For every $e=(u,w)$, fix one pair
$(a_e,b_e)\in R_{u,w}$ and for every $i,j\in[r]$,
if $x_e^{i,j}$ or $y_e^{i,j}$ belongs to $S$,
add $L_{u,a_e}^{i,j}$ and $R_{w,b_e}^{i,j}$.
Let $M'$ be the resulting set.
Since $M$ is a subset of $M'$, it remains a monitoring
edge-geodetic set of $G$. This proves (i).

We next show (ii).
A vertex of $M$ internal to an introduced path
requires at most two additions to obtain the
endpoints of that path.
Such a path has at most one endpoint of the
form $x_e^{i,j}$ or $y_e^{i,j}$, so it can
require at most two further additions.
Thus, each vertex of $M$ adds at most $4$ vertices to $M$, proving (ii).

We now prove (iii).
Fix $e=(u,w)$ and $i,j\in[r]$.
If $x_e^{i,j}$ or $y_e^{i,j}$ belongs to $S$,
the required pair was added explicitly.
Assume otherwise.

Choose two vertices of $M$ monitoring $t_e^{i,j}$,
and let $Q$ be a shortest path between them.
The endpoints of every subpath of $Q$ containing
$t_e^{i,j}$ also monitor this edge: otherwise,
an equally short replacement avoiding it would
contradict the choice of the original pair.

Let $P$ be the subpath between the first and
last vertices of $Q$ that belong to $\Gamma$
or are copies of vertices of $V_0$.
Every endpoint of $P$ inside a copy of $H$ belongs
to $S$. Indeed, if it is an endpoint of $Q$,
it already belongs to $M$.
Otherwise, the corresponding endpoint of $Q$
is internal to an introduced path, and $P$
starts or ends at an endpoint of that path.
Both endpoints of that introduced path were
added to $S$.

Consider the maximal subpath of $P$ containing
$t_e^{i,j}$ and lying in $H^{i,j}$, and denote
its endpoints by $p^{i,j}$ and $q^{i,j}$.
Their corresponding pair $\{p,q\}$ belongs
to $D$, since a shortest path in $H$ not including
$t_e$ would give an equally short replacement
avoiding $t_e^{i,j}$.

If $P$ continues outside this copy on both
sides, the portion between the nearest vertices
of $\Gamma$ on either side has length at least
seven: it uses two paths of length three and
$t_e^{i,j}$.
The added path of length six between these
vertices contradicts shortestness.

Suppose that exactly one endpoint of $P$ lies in
$H^{i,j}$, and denote it by $p^{i,j}$.
This vertex belongs to $S$ and is neither
$x_e^{i,j}$ nor $y_e^{i,j}$.
Starting from $p^{i,j}$, the path $P$ traverses
$t_e^{i,j}$ and leaves the copy through $q^{i,j}$,
then reaches a vertex $g\in\Gamma(q)$.
This requires at least six edges inside the copy
and three further edges to reach $g$.
However, concatenating $P(p^{i,j},\rho_{p,i})$
and $P(\rho_{p,i},g)$ gives a path of length $9$
avoiding $t_e^{i,j}$.
This path is no longer than the corresponding
subpath of $P$, contradicting the fact that
$p^{i,j}$ and $g$ monitor $t_e^{i,j}$.

Consequently, $P$ lies entirely in $H^{i,j}$,
and both its endpoints belong to $S$.
Since their corresponding pair belongs to $D$,
there is a path of length $12$ between them
through $\Gamma$ avoiding $t_e^{i,j}$.
Their monitoring this edge therefore forces
$P$ to have length at most $11$.

Neither endpoint is $x_e^{i,j}$ or $y_e^{i,j}$,
so $P$ contains a full path of length $5$
on each side of $t_e^{i,j}$.
These paths and the edge already account for
$11$ edges.
Thus, its endpoints are $L_{u,a}^{i,j}$ and
$R_{w,b}^{i,j}$ for some $a\in A$ and $b\in B$.
If $(a,b)\notin R_{u,w}$, the introduced path
of length $5$ between them would be shorter.
Hence $(a,b)\in R_{u,w}$, and both required
vertices belong to $S\subseteq M'$.
\end{proof}

We now use \cref{lem:all-edges-monitored-leaves,lem:normalize-meg-set}
to show that a polynomial-time $\alpha$-approximation
algorithm for \megset would yield a polynomial-time
$7\alpha$-approximation algorithm for \symmlabelcover,
thereby proving the following theorem.

\begin{theorem}
\label{thm:meg-set-inapproximability}
For every constant $0<\varepsilon<1$, \megset
admits no polynomial-time
$\left(2^{\log^{1-\varepsilon} n}\right)$-approximation algorithm,
unless
$\classNP\subseteq
\mathrm{DTIME}(n^{\polylog n})$.
\end{theorem}

\begin{proof}
Suppose that a polynomial-time
$\left(2^{\log^{1-\varepsilon} n}\right)$-approximation algorithm
for \megset exists.
Given a \symmlabelcover instance $\mathcal{I}$,
construct $G$. Let $C^*$ be the cost of an optimal
solution to $\mathcal{I}$, and let $M^*$
be an optimal monitoring edge-geodetic
set of $G$.
Let $\widehat{M}$ contain $F$ and, in every copy
$H^{i,j}$, the vertices corresponding to the labels
selected by an optimal solution of $\mathcal{I}$.
By \cref{lem:all-edges-monitored-leaves}, edges
outside $T_G$ are monitored by $\widehat{M} \cap L$.

For every $e=(u,w)$ and every $i,j\in[r]$,
the set $\widehat{M}$ contains vertices
$L_{u,a}^{i,j}$ and $R_{w,b}^{i,j}$ with
$(a,b)\in R_{u,w}$.
Their path containing $t_e^{i,j}$ has length eleven.
Any other path inside the copy uses at least three
paths of length five, while any path leaving the
copy has length at least twelve.
Hence, this pair monitors $t_e^{i,j}$.
Thus, $\widehat{M}$ is a monitoring edge-geodetic
set of $G$, and
$|M^*| \le |\widehat{M}| \le |F|+kC^*$.

Run the assumed algorithm on $G$, obtaining a set
$M$ with $|M|\le\alpha |M^*|$, where
$\alpha=2^{(\log n)^{1-\varepsilon}}$.
By \cref{lem:normalize-meg-set}, we can compute
in polynomial time a set $M'\subseteq V(G)$ such that
(i) $M'$ is a monitoring edge-geodetic set of $G$;
(ii) $|M'|\le5|M|$; and
(iii) for every $e=(u,w)\in E$ and every $i,j\in[r]$,
there exists $(a,b)\in R_{u,w}$ such that
$L_{u,a}^{i,j},R_{w,b}^{i,j}\in M'$.

Choose a copy $H^{i,j}$ minimizing the number of
vertices of $M'$ of the forms $L_{u,a}^{i,j}$ and
$R_{w,b}^{i,j}$.
Define $\phi(u)$ to contain precisely those labels
$a$ for which $L_{u,a}^{i,j}\in M'$, and define
$\psi(w)$ analogously using $R_{w,b}^{i,j}$.
Property~(iii) guarantees that $(\phi,\psi)$ is a solution for $\mathcal{I}$.

The cost of this solution equals the number of
selected vertices of these two forms in $H^{i,j}$.
Since the $k$ copies are vertex-disjoint, the total
number of such vertices belonging to $M'$ across
all copies is at most $|M'|$.
The chosen copy contains no more than their average
number, so $\cost(\phi,\psi)\le |M'|/k$.

Every monitoring edge-geodetic set contains all leaves
of the graph~\cite{Foucaud}, so $|F| \le |M^*|$.
By construction, $|F|=2r|V_0|\le2rh=k/2$.
Moreover, every feasible solution to the nonempty
instance $\mathcal{I}$ selects at least one label
for a vertex in $U$ and one for a vertex in $W$,
and therefore $C^*\ge2$.
Using $|M'|\le5|M|$, $|M|\le\alpha |M^*|$,
$|M^*| \le |F|+kC^*$, and $\frac{C^*}{4} \ge \frac{1}{2}$ we obtain:
\[
  \cost(\phi,\psi)
  \le \frac{|M'|}{k}
   \le \frac{5|M|}{k}
   \le \frac{5\alpha |M^*|}{k}
  \le \frac{5\alpha(|F|+kC^*)}{k}
  \le 5\alpha\left(C^*+\frac12\right)
  \le \frac{25}{4}\alpha C^*
   \le 7\alpha C^*.
\]

It remains to verify the size and running time
of the reduction.
The graph $H$ has $O(s^4)$ vertices and edges,
while $|V_0|=O(s^2)$.
Hence $k=O(s^8)$ and $|D|=O(s^4)$.
The $k$ copies of $H$ and the
$O(k|V_0|+k|D|)$ added paths contribute
$O(s^{12})$ vertices and edges in total,
including the vertices in $\Gamma$ and their
leaves, since all path lengths are constant.
The set $D$ can be computed by testing whether
deleting an edge of $T$ increases the distance
between a pair of vertices of $V_0$.

Thus both building the \megset instance, and recovering $(\phi, \psi)$ from a solution require polynomial time.
Since $n=O(s^{12})$, the approximation ratio
obtained for \symmlabelcover is at most
$
  7\cdot2^{(\log n)^{1-\varepsilon}}
  =2^{O((\log s)^{1-\varepsilon})}
  =o\!\left(2^{(\log s)^{1-\varepsilon/2}}\right)$.
For sufficiently large $s$, this contradicts
\cref{thm:symmetric-label-cover-hardness},
applied with $\delta=\varepsilon/2$, unless
$\classNP\subseteq
\mathrm{DTIME}(n^{\polylog n})$.
\end{proof}

\section{\texorpdfstring{$\classNP$}{NP}-hardness of \megset on Planar Graphs}
\label{sec:megset-planar-hardness}

In \cite{BCFL2025}, it was shown that \megset is $\classNP$-hard on $1$-apex graphs, and the complexity of \megset on planar graphs was explicitly left as an open problem.
In this section, we show that the problem remains $\classNP$-hard even in this case.

We give a reduction from the vertex cover problem restricted to cubic planar graphs (\vcp for short), which is
known to be $\classNP$-hard \cite{DinitzKatzKrakovski2009}. Let $G=(V,E)$ be an instance of \vcp with $n=|V|$ and a consider a straight-line planar embedding of $G$ (such an embedding can be found in polynomial time \cite{schnyder1990embedding}).
For each vertex $v$, we designate three \emph{ports}, labeled $0,1,2$ in clockwise order around $v$, to which its incident edges are attached. Each edge $e=(u,v)$ connects a port of $u$ with a port of $v$, and each port is connected to exactly one edge.

\begin{figure}
  \centering
  \includegraphics[width=.9\textwidth]{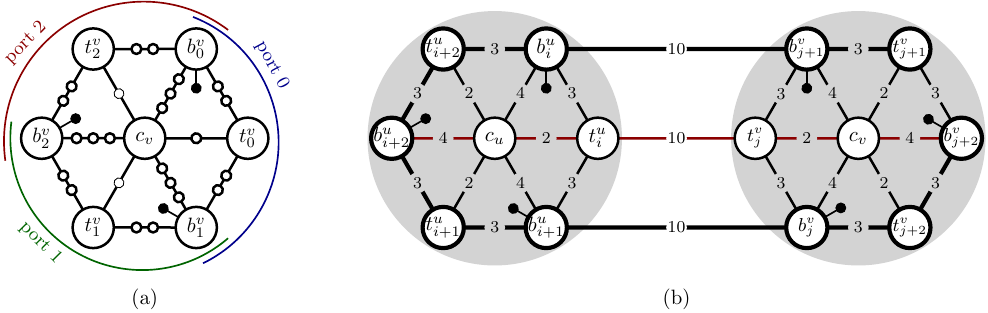}
  \caption{(a) The gadget used to replace a generic vertex $v$ of $G$. The portions of the outher cycle containing vertices $\{ b_i^v, t_i^v, b_{i+1}^v \}$ corresponding to port $i \in \{0,1,2\}$ are marked in blue, green, and red. (b) The replacement for an edge $(u,v) \in E(G)$ conencting port $i$ of $u$ with port $j$ of $v$. To reduce clutter, an edge labelled $\ell$ represents a path of length $\ell$ between its endvertices. The path $\pi_e$ used in the proof of \Cref{lem:from_vc_to_meg_set} is highlighted in red. The vertices in the set $X$ used in the proof of \Cref{lem:vc_intersects_gadgets} are shown in bold. The bold edges provide shortest paths from $b_{i+2}^u$ to the vertices in $X$.}
  \label{fig:planar_replacement}
\end{figure}

Our instance $H$ of \megset is obtained by local replacements of the 
vertices and of the edges of $G$.
More precisely, we replace each vertex $v$ of $G$ with the \emph{gadget} shown in Figure~\ref{fig:planar_replacement}~(a).
This gadget is a planar graph containing an outer cycle of $18$ vertices. The vertices encountered at every third position are named $b_{0}^v, t_{0}^v, b_{1}^v, t_{1}^v, b_{2}^v$, and $t_{2}^v$, in clockwise order. A center vertex $c_v$ is connected to each $b_i^v$ (resp.\ $t_i^v$) by a path of length $4$ (resp.\ $2$).
To lighten the notation, we write $b_{i}^v$ in place of $b_{i \bmod 3}^v$ so that $b_{3}^v$ coincides with $b_{0}^v$ and $b_{4}^v$ coincides with $b_{1}^v$.
Each of $b_{i}^v$ for $i=0,\dots,2$ is connected to a leaf $\ell_{i}^v$ via the edge $(b_{i}^v, \ell_{i}^v)$. Vertices $b_{i}^v$, $t_{i}^v$, and $b_{i+1}^v$ correspond to the $i$-th port of $v$. We denote by $S_v$ the set of vertices of $v$'s gadget.
Each edge $e=(u,v)$ of $G$ is replaced by three paths in $H$, as shown in Figure~\ref{fig:planar_replacement}~(b). In particular, if $e$ is connected to port $i$ of $u$ and port $j$ of $v$, then $H$ contains: 
(i) a path from $t_{i}^u$ to $t_{j}^v$ of length $10$,
(ii) a path from $b_{i}^u$ to $b_{j+1}^v$ of length $10$, and 
(iii) a path from $b_{i+1}^u$ to $b_j^v$ of length $10$. We refer to the path in (i) as a $t$-path and to the paths in (ii) and (iii) as $b$-paths. We denote by $P_{u,v}$ the set of internal vertices of the $b$- and $t$-paths connecting $u$ and $v$, and by $L$ the set of all vertices of degree $1$ in $H$.
\Cref{fig:planar_example} shows the graph $H$ obtained by applying the above construction to a small graph $G$.

\begin{figure}
  \centering
  \includegraphics[width=0.95\textwidth]{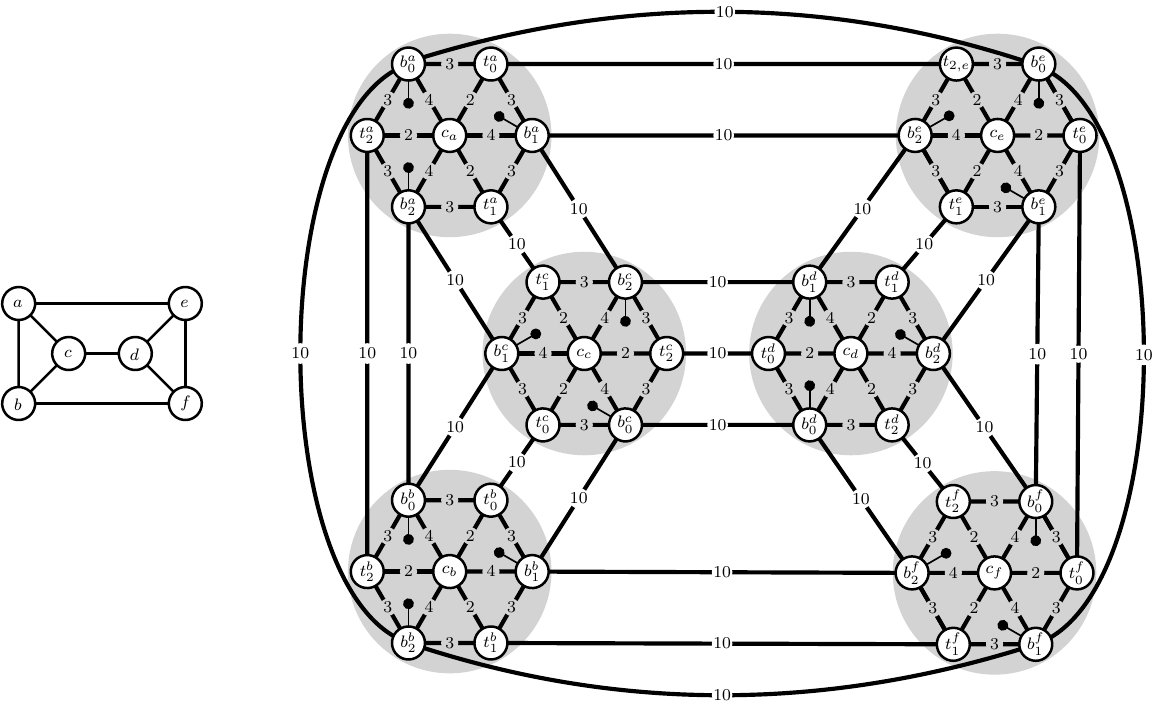}
  \caption{Left: the graph $G$ of an instance of vertex cover. Right: the graph $H$ obtained by applying our reduction to $G$. Vertex gadgets have a gray background.}
  \label{fig:planar_example}
\end{figure}

\begin{lemma}\label{lem:from_vc_to_meg_set}
If $C\subseteq V$ is a vertex cover of $G$ then
$M = \{ c_v : v \in C \} \cup L$ is a monitoring edge-geodetic set of $H$.
\end{lemma}
\begin{proof}
  For a vertex $v \in V$, the unique shortest path in $H$ from $b_i^v$ to $b_{i+1}^v$ follows a portion of the outer cycle of $v$'s gadget that contains $t_i^v$ and has length $6$. Hence, all edges in the unique shortest path from $\ell_i^v$ to $\ell_{i+1}^v$ are monitored by $\{\ell_i^v,\ell_{i+1}^v\}$. Thus, all edges incident to vertices in $L$ and all edges of the outer cycles in the vertex gadgets are monitored by $M$.

  All edges of $H$ that belong to some $b$-path with endvertices $b_{i}^u$ and $b_{j}^v$ for some choice of $u,v \in V$ and $i,j \in \{0,1,2\}$, are monitored by the pair $\{\ell_{i}^u, \ell_{j}^v\}$.

  For a generic edge $e=(u,v) \in E$ connecting port $i$ of vertex $u$ to port $j$ of vertex $v$, define $\pi_e$ as the path obtained by concatenating:
  (i) the shortest path from $b_{i+2}^u$ to $c_u$ of length $4$ in $u$'s gadget,
  (ii) the shortest path from $c_u$ to $t_{i}^u$ of length $2$  in $u$'s gadget,
  (iii) the $t$-path from $t_{i}^u$ to $t_{j}^v$ of length $10$,
  (iv) the shortest path from $t_{j}^v$ to $c_v$ of length $2$ in $v$'s gadget, and
  (v) the shortest path from $c_v$ to $b_{j+2}^v$ of length $4$  in $v$'s gadget.
   See \Cref{fig:planar_replacement}~(b). 

  Observe that all edges that we have not already shown to be monitored belong to some path $\pi_e$. Hence, it suffices to show that all edges on $\pi_e$ are monitored.
  Assume w.l.o.g.\ that $u \in C$, which implies $c_u \in M$.
  The edges in subpath (i) of $\pi_e$ are monitored by $\{\ell_{i+2}^u, c_u\}$,
  while the remaining edges of $\pi_e$ are monitored by $\{c_u, \ell_{j+2}^v\}$.
  To see the latter, notice that the concatenation of the subpath of $\pi_e$ from $c_u$ to $b_{j+2}^v$ with the edge $(b_{j+2}^v, \ell_{j+2}^v)$ has length $19$.
  Since any simple path containing at least one vertex not in $S_u \cup S_v \cup P_{u,v}$ has length at least $20$, we can consider the subgraph of $H$ induced by $S_u \cup S_v \cup P_{u,v}$.
  In such a subgraph, all paths from $c_u$ to $\ell_{j+2}^v$ have a length of at least $4 + 10+ 6 = 20$ (see \Cref{fig:planar_replacement}~(b)).
\end{proof}

\begin{lemma}\label{lem:vc_intersects_gadgets}
Let $M$ be a monitoring edge-geodetic set of $H$.
Then, for every $(u,v) \in E$, the set $M$ intersects at least one of $S_u \setminus L$, $S_v \setminus L$, and $P_{u,v}$.
\end{lemma}
\begin{proof}
For $z \in V$, let $C_z = \{ t_i^z : i \in \{0,1,2\} \} \cup \{ b_{i}^z : i \in \{0,1,2\} \}$.
Let $i$ and $j$ be such that $(u,v)$ connects port $i$ of vertex $u$ to port $j$ of vertex $v$ in $G$. Let $e = (b_{i+2}^u, b')$  be the first edge on the shortest path connecting $b_{i+2}^u$ to $c_u$ in $u$'s gadget.
Let $X = (C_u \setminus \{ t_{i}^u \}) \cup (C_v \setminus \{ t_{j}^v\})$.

Suppose, for a contradiction, that $M$ does not intersect any of $S_u \setminus L$, $S_v \setminus L$, and $P_{u,v}$. Since $X \subseteq (S_u \setminus L) \cup (S_v \setminus L)$, we have $M \cap X = \emptyset$.

Let $\{a,a'\} \in \binom{M}{2}$ be a pair of vertices that monitors $e$ in $H$, let $\pi$ be any shortest path from $a$ to $a'$, and assume w.l.o.g.\ that $\pi$ encounters $b_{i+2}^u$ before $b'$.
Let $x$ be the first vertex of $X$ traversed by the subpath of $\pi$ from $b'$ to $a'$ (such a vertex must exist: $a' \in M$ and $M \cap X = \emptyset$, while $b' \notin X$ and $a'$ is disconnected from $b'$ in $H-X$).
The pair $\{b_{i+2}^u, x\}$ monitors $e$, and the subpath of $\pi$ from $b_{i+2}^u$ to $x$ is a shortest path in $H$ that contains only vertices in the subgraph $H' = H[S_u \cup S_v \cup P_{u,v}]$. Then, all shortest paths from $b_{i+2}^u$ to $x$ in $H'$ must traverse $e$.
This is contradictory since, for every $y \in X$, there is a shortest path $\pi_y$ in $H'$ from $b_{i+2}^u$ to $y$ such that $\pi_y$ avoids $e$ (see \Cref{fig:planar_replacement}~(b)). 
\end{proof}

\begin{theorem}
  \megset is $\classNP$-hard even when instances are restricted to planar graphs with degenaracy $2$ and maximum degree $6$.
\end{theorem}
\begin{proof}
By construction, $|L|=3n$. We show that $G$ admits a vertex cover of size at most $k$ if and only if $H$ has a monitoring edge-gedetic set of size at most $k + 3n$.

By \Cref{lem:from_vc_to_meg_set}, if there exists a vertex cover $C$ of $G$ of size at most $k$, then $M = \{c_v: v \in C\} \cup L$ is a monitoring edge-geodetic set of $H$ of size at most $k + |L| = k + 3n$.
 
Suppose now that there exists a monitoring edge-geodetic set $M$ of $H$ of size at most $k + 3n$. Observe that $L \subseteq M$ since a leaf vertex must necessarily belong to $M$ for its sole incident edge to be monitored.

We build a set $M'$ in two steps: (i) for each vertex $v \in V$ such that $M$ intersects $S_v \setminus L$, we add $v$ to $M'$; (ii) for each edge $(u,v) \in E$ such that $M$ intersects $P_{u,v}$, we choose one vertex $z \in \{u, v\}$ and add $z$ to $M'$ (if $z$ is already in $M'$, we do nothing).
By \Cref{lem:vc_intersects_gadgets}, for every edge $(u,v) \in E$, $M'$ contains at least one of $u$ and $v$.
Indeed, if $M \cap (S_u \setminus L) \neq \emptyset$, then $u \in M'$ by step (i); if $M \cap (S_v \setminus L) \neq \emptyset$, then $v \in M'$ by step (i); and if $M \cap P_{u,v} \neq \emptyset$, then some $z \in \{u,v\}$ belongs to $M'$ by step (ii).
Hence, $M'$ is a vertex cover for $G$.
Moreover, the sets $S_v \setminus L$ for $v \in V$ and the sets $P_{u,v}$ for $(u,v) \in E$ partition $V(H) \setminus L$. Each set in this partition that intersects $M$ contributes at most one vertex to $M'$, while the other sets contribute $0$.
Therefore, $|M'| \le |M \setminus L| = |M| - |L| \le k + 3n - 3n = k$.

The claim follows by observing that the graph $H$ has maximum degree $6$, and that deleting all vertices of degree at most $2$ results in a graph with no edges.
\end{proof}

\section{Lower Bounds via Subsetwise Distance Preservers}
\label{sec:preservers}

In this section, we relate monitoring edge-geodetic sets to subsetwise
distance preservers and use this connection to derive lower bounds
that depend on the density of the input graph. 

Let $G=(V,E)$ be a graph and let
$P\subseteq\binom{V}{2}$. A spanning subgraph $H$ of $G$ is a
\emph{pairwise distance preserver} of $G$ w.r.t.\ the pairs $P$ if
$ d_H(x,y)=d_G(x,y)$ for every $\{x,y\} \in P$.
If $P = \binom{S}{2}$ for some $S \subseteq V$, then $H$ is called a \emph{subsetwise distance preserver} of $G$ w.r.t.\ $S$.

\begin{lemma}
\label{lem:meg-set-preserver}
Let $M$ be a monitoring edge-geodetic set of $G=(V,E)$. 
The only subsetwise distance preserver of $G$ w.r.t.\ $M$ is $G$ itself.
\end{lemma}
\begin{proof}
Let $H\subseteq G$ be a subsetwise distance preserver of $G$ w.r.t.\ $M$.
By definition we have $V(H) = V$ and $E(H) \subseteq E$, so it suffices to show that $E \subseteq E(H)$.
Focus on a generic edge $e\in E$. Since $M$ is a monitoring edge-geodetic
set, there exists a pair $\{x,y\}\in \binom{M}{2}$ that monitors $e$ in $G$. Then, every shortest path from $x$ to $y$ in $G$ contains $e$.
Since $H$ preserves the distance between $x$ and $y$, we must have $e\in E(H)$.
\end{proof}

\begin{corollary}
  \label{cor:meg-set-preserver}
  Let $f(n, |S|)$ be a function such that, for any graph $G$ with $n$ vertices and any $S \subseteq V(G)$,
  there exists a subsetwise distance preserver of $G$ w.r.t.\ $S$ with
  at most $f(n, |S|)$ edges.
  Then, every very monitoring edge-geodetic set of a graph with $n$ vertices and
  $m$ edges has size at least $g(n,m) = \min\{s \in \mathbb{N} : f(n,s) \ge m\}$.
\end{corollary}
\begin{proof}
Let $G$ be a graph with $n$ vertices and $m$ edges, and let $M$ be
any monitoring edge-geodetic set of $G$. By the assumption on $f$,
there exists a subsetwise distance preserver of $G$ w.r.t.\ $M$ with
at most $f(n,|M|)$ edges. By \cref{lem:meg-set-preserver}, this
preserver must be $G$ itself, so $f(n,|M|) \ge m$. Therefore,
$|M| \ge g(n,m)$ by the definition of $g$.
\end{proof}

Bodwin and Vassilevska Williams~\cite{BodwinWilliams2016} proved
that graphs with $n$ vertices admit pairwise distance preservers w.r.t.\ $P$ with $O(n^{2/3} |P|^{2/3}+n |P|^{1/3})$ edges, and hence subsetwise distance preserves w.r.t.\ $S$ with $f(n, |S|) = O(n^{2/3} |S|^{4/3} + n |S|^{2/3})$ edges.
Combining their result with \Cref{cor:meg-set-preserver} gives the following lower bound:
\begin{theorem}
\label{thm:meg-set-preserver-lower-bound}
All monitoring edge-geodetic sets of an $n$-vertex, $m$-edge graph contain
$\Omega\left(
        \min\left\{
            \frac{m^{3/4}}{n^{1/2}},
            \left(\frac{m}{n}\right)^{3/2}
        \right\}
    \right)$ vertices.
In particular, if $m=\Theta(n^2)$, then $|M|=\Omega(n)$.
\end{theorem}
\begin{proof}
From \cite{BodwinWilliams2016}, we know that all graphs $G$ with $n$ vertices admit subsetwise preservers w.r.t.\ $S \subseteq V(G)$ of size at most $f(n, |S|) \le c n^{2/3} |S|^{4/3}+ c n |S|^{2/3}$, for some constant $c > 0$.

Hence, \Cref{cor:meg-set-preserver}
shows a lower bound on the size of a monitoring edge-geodetic set for graphs with $n$ veritces and $m$ edges 
of at least the minimum integer $s$ that satisfies 
$f(n,s) = c n^{2/3} s^{4/3}+ c n s^{2/3} \ge m$. Then, $n^{2/3} s^{4/3} \ge \frac{m}{2c}$, or $n s^{2/3} \ge \frac{m}{2c}$ (possibly both).
Hence $s \ge \frac{m^{3/4}}{(2c)^{3/4} n^{1/2}} = \Omega(\frac{m^{3/4}}{n^{1/2}})$ or $s \ge \frac{m^{3/2}}{(2c)^{3/2} n^{3/2}} = \Omega( \left( \frac{m}{n}\right)^{3/2} )$, which implies the claim.
\end{proof}

\subparagraph*{Generative AI Disclosure.}
Generative AI tools were used to assist with language
editing and improving the clarity of the manuscript.
All scientific content, mathematical results, and proofs
were developed, reviewed, and verified by the authors,
who take full responsibility for the final manuscript.



\bibliography{bibliography}

\end{document}